\documentclass{amsart}
\usepackage{amsmath,amssymb,amsthm,amsfonts}
\usepackage{pstricks,pstricks-add,pst-node,tikz}
\newtheorem{Theorem}{Theorem}[section]

\newtheorem{Conjecture}[Theorem]{Conjecture}

\newtheorem{Lemma}[Theorem]{Lemma}
\newtheorem{Corollary}[Theorem]{Corollary}
 
\newcommand{\torr}{\stackrel{R}{\rightarrow}}
\newcommand{\vs}{\vspace{0.5cm}}
\begin{document}
\date{\today}
\title{Rainbow Induced Subgraphs in Replication Graphs}
\author{Marek Szyku{\l}a and Andrzej Kisielewicz}
\address{University of Wroc\l aw\\
Department of Mathematics and Computer Science\\
pl. Grunwaldzki 2, 50-384 Wroc\l aw, Poland}
\email{marek.szykula@ii.uni.wroc.pl, andrzej.kisielewicz@math.uni.wroc.pl}

\begin{abstract}
A graph $G$ is called a replication graph of a graph $H$ if $G$ is obtained from $H$ by replacing vertices of $H$ by arbitrary cliques of vertices and then replacing each edge in $H$ by all the edges between corresponding cligues. For a given graph $H$ the $\rho_R(H)$ is the minimal number of vertices of a replication graph $G$ of $H$ such that every proper vertex coloring of $G$ contains a rainbow induced subgraph isomorphic to $H$ having exactly one vertex in each replication clique of $G$. We prove some bounds for $\rho_R$ for some classes of graphs and compute some exact values. Also some experimental results obtained by a computer search are presented and conjectures based on them are formulated.

\end{abstract}

\maketitle

%%%%%%%%%%%%%%%%%%%%%%%%%%%%%%%%%%%%%%%%%%%%%%%%%%%%%%%%%%%%%%%%%%%%%%%%%%%%%%%%%
\section{Introduction}

Rainbow induced subgraphs have been considered in many papers in connection with various problems of extremal graph theory. They have been considered both for edge-colorings and vertex-colorings, and both in terms of existence or in terms of avoiding (see \cite{AC,AI,AM,AS,KMSV,KT,KL}). Our special motivation comes from research in on-line coloring (see \cite{BCP,Ma}), where the base for some algorithms is the existence of rainbow anticliques to force a player to use a new color. In particular, in \cite{Ma}, a problem has been formulated to estimate the minimal number of moves in the game considered one needs to force the appearance of a rainbow copy of a fixed graph $H$ in a fixed class of graphs $C$. 

This paper is continuation of \cite{KS}, where the number $\rho(H)$ being the minimum order of a graph $G$ such that every proper vertex coloring of $G$ contains a rainbow induced subgraph isomorphic to $H$ was introduced. It turned out that in certain situations, especially for paths, more natural and easier to handle is the number $\rho_R(H)$ referring to replication graphs defined as follows.

Given a graph $H$ and a vertex $v\in H$, we construct a graph $H'$ by adding a new vertex $v'$ and edges joing $v'$ with $v$ and all the neighbours of $v$ in $H$. Aany graph $G$ obtained fro $H$ by a series of such constructions is called a replication graph of $G$. Note that vertices in $G$ corresponding to a vertex $v\in H$ form a clique. 

If $G$ is a replication graph of $H$ such that in each proper vertex coloring of $G$ there exists a rainbow induced subgraph $H$ \textbf{having exactly one vertex in each of the cliques} $K_i$ corresponding to a vertex $h_i$ in $H$, then we write $G\torr H$. By $\rho_R(H)$ we denote the minimal number of vertices in any replication graph $G$ of $H$ satisfying $G \torr H$.

In this paper we provide some bounds for $\rho_R$ for certain classes of graphs. Following problems formulated in \cite{KS} our main interest is in paths. In addition, the exact value of $\rho_R$ for a double star is computed. We present also some experimental results obtained by computer search for small paths and conjectures based on them.

%%%%%%%%%%%%%%%%%%%%%%%%%%%%%%%%%%%%%%%%%%%%%%%%%%%%%%%%%%%%%%%%%%%%%%%%%%%%%%%%%
\section{The Hall's type theorem}

We start from the following more general result. 

\begin{Theorem}
Let $H$ be a graph and $G$ be its replication graph. The following conditions are equivalent:
\begin{itemize}
\item For each subset $S$ of the vertices of $H$, the chromatic number of the subgraph of $G$ induced by the replication cliques corresponding to the vertices from $S$ (denoted by $G[S]$) is at least $|S|$.
\item $G \torr H$.
\end{itemize}
\end{Theorem}
\begin{proof}
Let the first condition be satisfied. We take any proper vertex coloring and we will show that there is a rainbow $H$ having one vertex in each replication clique. We define a bipartite graph consisted of the set of vertices of $H$ and the set of colors used in the coloring. A vertex $v$ is connected with all colors used in the replication clique obtained from $v$. Since each subset of $k$ vertices is connected with at least $k$ colors then by Hall's theorem we obtain that there exist matching such that each vertex is matched with unique color which is connected to it. So the matching defines the colors of the vertices in replication cliques and selecting a vertex from each replication clique of these colors gives the rainbow induced subgraph $H$.

Conversely, if for some subset $S$ the chromatic number of the subgraph is less than $|S|$ then we could get a coloring which uses less than $|S|$ colors for these cliques and extend it to the whole graph $G$ obtaining a coloring from which we cannot select rainbow vertices of $H$ from each replication clique.
\end{proof}

\begin{Lemma}[\cite{LO}]
A replication of a perfect graph is perfect.
\end{Lemma}

Thus for perfect graphs we could consider the maximal clique instead of the chromatic number of the subgraph induced by a subset. Paths are perfect and so the replication graphs of them. The maximal clique in the replication graph of path is always formed by two adjacent replication cliques.

%%%%%%%%%%%%%%%%%%%%%%%%%%%%%%%%%%%%%%%%%%%%%%%%%%%%%%%%%%%%%%%%%%%%%%%%%%%%%%%%%
\section{The upper bound for paths}

\begin{Theorem}$$\rho_R(P_n) \le \left\{\begin{array}{ll}
n^2/4+n^2/16+n/2 & \mbox{if $n \equiv 0 \mod 4$}\\
n^2/4+n^2/16+3n/8+5/16 & \mbox{if $n \equiv 1 \mod 4$}\\
n^2/4+n^2/16+n/2-1/4 & \mbox{if $n \equiv 2 \mod 4$}\\
n^2/4+n^2/16+3n/8+1/16 & \mbox{if $n \equiv 3 \mod 4$}\\
\end{array}\right.$$\end{Theorem}
\begin{proof}

We will define a suitable $G$ as a sequence of orders of replication cliques of $P_n$, and show that each $k$-subset of this sequence contains a number at least $k$ or a pair of consecutive numbers in the sequence of sum at least $k$.

\begin{enumerate}
\item Assume that $n \equiv 0 \mod 4$.

Define four sequences of length $n/4$ of numbers:
$$S_1=(1,2,...,n/4)$$
$$T_1=(n/2,n/2,...,n/2)$$
$$S_2=(n/4+1,n/4+2,...,n/2)$$
$$T_2=(n/4,n/4,...,n/4,n/2)$$

Then $G$ is defined by alternately taking the numbers from $S_1,T_1,S_2,T_2$ in that order, so the first number is the first of $S_1$, the second is the first of $T_1$, ..., the fifth is the second of $S_1$ and so on.

For example for $n=16$:
$$S_1=(1,2,3,4)$$
$$T_1=(8,8,8,8)$$
$$S_2=(5,6,7,8)$$
$$T_2=(4,4,4,8)$$
$$S=(1,5,2,6,3,7,4,8)$$
$$T=(8,4,8,4,8,4,8,8)$$
$$G=(1,8,5,4,2,8,6,4,3,8,7,4,4,8,8,8)$$
It could be seen in this way:
$$\begin{array}{llllllllllllllllllll}
S_1&=&(&1& & & &2& & & &3& & & &4& & & &)\\
T_1&=&(& &8& & & &8& & & &8& & & &8& & &)\\
S_2&=&(& & &5& & & &6& & & &7& & & &8& &)\\
T_2&=&(& & & &4& & & &4& & & &4& & & &8&)\\
\\
G  &=&(&1&8&5&4&2&8&6&4&3&8&7&4&4&8&8&8&)
\end{array}$$
    
Consider a $k$-subset of numbers of $G$. By a value of the subset we mean maximum over numbers from the subset or sums of pairs of consecutive numbers in the sequence which are both in the subset. The subset contains $s_1$ numbers from $S_1$, $s_2$ numbers from $S_2$, $t_1$ numbers from $T_1$ and $t_2$ numbers from $T_2$.

Assume for the contrary that the value of the subset is less than $k$. There are the following cases:
\begin{enumerate}
\item If $k \le n/4$ then all numbers in the subset must be from $S_1$, but there are only $k-1$ numbers less than $k$, so this is impossible.

\item If $n/4 < k \le n/2$ there are not numbers from $T_1$.

Consider a set $X_i$ for $1 \le i \le n/4-1$ which contains $i$-th numbers from $S_2$ and $T_2$ and $(i+1)$-th number from $S_1$, but only those which are also in the $k$-subset. In each $X_i$ there could be at most two numbers because there cannot be a pair between $S_2$ and $T_2$. There could be at most one number if $i \ge k-n/4$ because there cannot be a number from $S_2$ and cannot be a pair between $T_2$ and $S_1$. In summary there are at most
$$2*(k-n/4-1)+((n/4-1)-(k-n/4)+1) = k-2$$
numbers in these $X_i$ sets.

Except the first number $1$ of the $S_1$ there could be only numbers from the sets $X_i$. So there is at most $k-1$ numbers in the subset, which is a contradiction.

\item If $n/2 < k \le n/2+n/4$ then we define $X_i$ for each $1 \le i \le n/4$ which contains $i$-th numbers from $S_1$, $T_1$, $S_2$, $T_2$ but only those which are also in the $k$-subset.

Observe that in each $X_i$ there cannot be both numbers from $T_1$ and $S_2$. Also there cannot be both from $S_1$ and $T_1$ or both from $S_2$ and $T_2$ if $i \ge k-n/2$. So for $i < k-n/2$ there can be at most $3$ numbers and for $i \ge k-n/2$ there can be at most $2$ numbers.

In summary in the $X_i$ sets there could be at most $$3*(k-n/2-1)+2*(n/4-(k-n/2)+1) = k-1$$ numbers. But every number from the subset would belong to some $X_i$ so this is a maximal size of the $k$-subset and a contradiction occurs.

\item If $n/2+n/4 < k \le n$ then we observe that there could be at most $k-n/2-n/4-1$ pairs between numbers from $T_1$ and $S_2$. So the maximal size of the $k$-subset is $$3*(n/4)+(k-n/2-n/4-1) = k-1.$$ This a contradiction.
\end{enumerate}

\item Assume that $n \equiv 1 \mod 4$.

We define sequences:
$$S_1=(1,2,...,(n+3)/4) \mbox{ of length $(n+3)/4$} $$
$$T_1=((n+3)/4,(n+3)/4,...,(n+3)/4) \mbox{ of length $(n-1)/4$} $$
$$S_2=((n+3)/4+1,(n+3)/4+2,...,(n+1)/2) \mbox{ of length $(n-1)/4$} $$
$$T_2=((n+1)/2,(n+1)/2,...,(n+1)/2,(n+3)/4) \mbox{ of length $(n-1)/4$} $$
And $G$ in similar way as before.

For example for $n=13$:
$$\begin{array}{lllllllllllllllll}
S_1&=&(&1& & & &2& & & &3& & & &4&)\\
T_1&=&(& &4& & & &4& & & &4& & & &)\\
S_2&=&(& & &5& & & &6& & & &7& & &)\\
T_2&=&(& & & &7& & & &7& & & &4& &)\\
\\
G  &=&(&1&4&5&7&2&4&6&7&3&4&7&4&4&)
\end{array}$$

\item Assume that $n \equiv 2 \mod 4$.

We define sequences:
$$S_1=(1,2,...,(n+2)/4) \mbox{ of length $(n+2)/4$} $$
$$T_1=(n/2,n/2,...,n/2) \mbox{ of length $(n+2)/4$} $$
$$S_2=((n+2)/4+1,(n+2)/4+2,...,n/2) \mbox{ of length $(n-2)/4$} $$
$$T_2=((n+2)/4,(n+2)/4,...,(n+2)/4) \mbox{ of length $(n-2)/4$} $$
And $G$ in similar way as before.

For example for $n=14$:
$$\begin{array}{llllllllllllllllll}
S_1&=&(&1& & & &2& & & &3& & & &4& &)\\
T_1&=&(& &7& & & &7& & & &7& & & &7&)\\
S_2&=&(& & &5& & & &6& & & &7& & & &)\\
T_2&=&(& & & &4& & & &4& & & &4& & &)\\
\\
G  &=&(&1&7&5&4&2&7&6&4&3&7&7&4&4&7&)
\end{array}$$

\item Assume that $n \equiv 3 \mod 4$.

We define sequences:
$$S_1=(1,2,...,(n+1)/4) \mbox{ of length $(n+1)/4$} $$
$$T_1=((n+1)/4,(n+1)/4,...,(n+1)/4) \mbox{ of length $(n+1)/4$} $$
$$S_2=((n+1)/4+1,(n+1)/4+2,...,(n+1)/2) \mbox{ of length $(n+1)/4$} $$
$$T_2=((n+1)/2,(n+1)/2,...,(n+1)/2) \mbox{ of length $(n-3)/4$} $$
And $G$ in similar way as before.

For example for $n=15$:
$$\begin{array}{llllllllllllllllllllll}
S_1&=&(&1& & & &2& & & &3& & & &4& & &)\\
T_1&=&(& &4& & & &4& & & &4& & & &4& &)\\
S_2&=&(& & &5& & & &6& & & &7& & & &8&)\\
T_2&=&(& & & &8& & & &8& & & &8& & & &)\\
\\
G  &=&(&1&4&5&8&2&4&6&8&3&4&7&8&4&4&8&)
\end{array}$$

\end{enumerate}
\end{proof}

%%%%%%%%%%%%%%%%%%%%%%%%%%%%%%%%%%%%%%%%%%%%%%%%%%%%%%%%%%%%%%%%%%%%%%%%%%%%%%%%%
\section{The lower bound for paths}

We will show that the exact value of $\rho_R(P_n)$ is above the simple bound $n^2/4$ by an $O(n^2)$ component.

\begin{Lemma}\label{lm1}Let $G$ be a minimal graph such that $G \torr P_n$, so the order of $G$ is $\rho_R(P_n)$. Let $$e = \sum_{ab} \min\{|A|,|B|\},$$ where each $ab$ is an edge of $P_n$ between vertices $a$ and $b$, $A$ and $B$ are replication cliques in $G$ obtained from the vertices $a$ and $b$. The following holds: $$|G| = \rho_R(P_n) \ge \frac{n(n+1)}{2} - e$$\end{Lemma}
\begin{proof}
We show that we could obtain from $G$ a graph $G'$ which is a replication graph of $A_n$ such that $G' \torr A_n$ by adding exactly $e$ vertices.

At first for each edge $ab$ in $P_n$ let us mark the smaller replication clique $A$ or $B$ (or any of them if they are equal). The cliques can be marked marked twice (by two edges), once (by one edge) or left unmarked. Note that we have one more vertex than the number of edges, so at least one clique is unmarked. If a clique is unmarked then its neighbor cliques must be marked.

For a graph $H$ which is a replication graph of a disjoint union of paths $U$, such that $G'' \torr U$, we define the procedure: Get a clique $A$ which is unmarked and which is connected with one or two neighbor cliques. Then increase number of vertices in $A$ by the sum of orders of the neighbor cliques, and remove connections between $A$ and them. Also we remove one mark from each of the neighbor cliques.

The obtained graph $H'$ is a replication graph of $U'$ which is $U$ without one or two edges. There is still an unmarked clique in $U'$ having a neighbor or $U'$ is just $A_n$. We show that $H' \torr U'$. Get any $k$-subset of $U'$. If the subset does not contain $A$ then it has the same chromatic number in $H'$ as in $H$, so it has at least $k$. If the subset contains $A$ then assume that it has a smaller chromatic number than $k$ in $H'$. With the fact in $H$ the subset has the chromatic number at least $k$ it must come from an induced clique between $A$ and one of its neighbor clique. But in $H'$ the clique $A'$ is at least as large as the induced clique and so $k$.

We use the defined procedure for $G$ and repeat it until we obtain a replication graph of $A_n$. For each mark on a clique $A$ we have added exactly $|A|$ vertices during the process, because we have not added vertices to marked cliques. So by a way in which we marked the cliques we have added exactly $e$ vertices. The order of the result graph must be at least $\rho_R(A_n) = \frac{n(n+1)}{2}$ so the lemma holds.
\end{proof}

%%%%%%%%%%%%%%%%%%%%%%%%%%%%%%%%%%%%%%%%%%%%%%%%%%%%%%%%%%%%%%%%%%%%%%%%%%%%%%%%%

\begin{Lemma}\label{lm2}Let $G$ be a replication graph of $P_n$ such that $G \torr P_n$ and $2|n$. Let $c,x$ be any numbers from $[0,1/2]$. If $$|G| \le (n/2)(n/2+1)+\frac{cn(cn-1)}{4}$$ then the following holds:
\begin{itemize}
\item Let $X$ be a subset of cliques which have orders at most $xn$. Then $$2xn-cn \le |X| \le 2xn.$$
\item Let $X'$ be a subset of cliques which have orders at least $xn$. Then $$n-2xn \le |X'| \le n-2xn+cn.$$
\item Let $X''$ be a subset of cliques which have orders at least $xn$ and at most $yn$. Then $$|X''| \ge 2n(y-x)-cn.$$
\end{itemize}
\end{Lemma}
\begin{proof}
Assume that $|G| < (n/2)(n/2+1)+\frac{cn(cn-1)}{4}$. If we sort the replication cliques ascending by the order and consider $i$-th replication clique then it must have order at least $\lceil i/2\rceil$. Otherwise if we get $i$-subset of the cliques smaller than $\lceil i/2\rceil$ then its the largest induced clique and so the chromatic number will be at most $(\lceil i/2\rceil-1)*2 \le i-1$.

If an $i$-th replication clique of order $k$ has more than $\lceil i/2\rceil$ vertices then we say that is has $k-i/2$ extra vertices. So we have exactly $(n/2)(n/2+1)$ non-extra vertices in $G$.

Consider a subset $X$ of cliques which have orders at most $xn$. The clique of order not larger than $xn$ can be at most $2xn$-th clique in our order, so $|X| \le 2xn$. If $|X| < 2xn-cn$ then the $i$-th cliques where $i = 2xn-cn,2xn-cn+1,...,2xn-1$ have at 
least $xn-\lceil(2xn-cn)/2\rceil,xn-\lceil(2xn-cn+1)/2\rceil,...,xn-\lceil(2xn-1)/2\rceil$ extra vertices respectively. So they have at least $cn/2-1,cn/2-2,...,0$ extra vertices respectively. In summary we have at least $(cn/2-1)(cn/2)/2 \ge \frac{cn(cn-1)}{4}$ extra vertices and it contradicts the assumption, so we are done in the first case.

Consider a subset $X'$ of cliques which have orders at least $xn$. So $|X'| \ge n-|X|$. By our bounds $|X'| \ge n-2xn$ and $|X'| \le n-(2xn-cn) = n-2xn+cn$. So we are done in the second case.

Consider a subset $X''$ of cliques which have orders at least $xn$ and at most $yn$. Then $|X''| \ge |X'|-|Y'|$ where $X'$ contains cliques of orders at least $xn$ and $Y'$ contains cliques of orders at least $yn$. So $|X''| \ge n-2xn-(n-2yn+cn) = 2n(y-x)-cn$ and we are done in the third case.
\end{proof}

%%%%%%%%%%%%%%%%%%%%%%%%%%%%%%%%%%%%%%%%%%%%%%%%%%%%%%%%%%%%%%%%%%%%%%%%%%%%%%%%%

\begin{Theorem}\label{th1}Let $G$ be a replication graph of $P_n$ such that $G \torr P_n$ and $2|n$. Let $a',a,b,c,d$ be the numbers from $[0,1/2]$ such that:
\begin{enumerate}
\item $b < a < a'$.
\item $a+\frac{3}{2}c+d \le a'$.
\item $a+\frac{3}{2}c+d \le 2b$.
\end{enumerate}
The following holds:
$$|G| \ge (n/2)(n/2+1)+\min\left\{\frac{cn(cn-1)}{4},\frac{d(a-b)}{4}n^2\right\}.$$
\end{Theorem}
\begin{proof}
Get a subset $M$ of cliques of orders in $[an,a'n]$. Let $N$ be a subset of cliques of orders at most $bn$. Let $M_N$ be a subset of cliques of orders at most $bn$ which are also neighbors of some clique from $M$.

We could bound $e$ defined in \ref{lm1}:
$$e \le \rho_R(P_n) - |M_N|\frac{a-b}{2}n.$$
If for each edge we take the order of the clique obtained from the left vertex into $e$ then $e < |G|$. Now for these edges which have the left vertex producing clique from $M_N$ and right vertex producing clique from $N$ we could take the order of the right clique instead of the left. For each such edge we get the number of vertices smaller at least by $(a-b)n$. At least half of the cliques from $M_N$ have a neighbor from $M$ on the left side or on the right side. If it is the first case then we have at least $|M_N|/2$ such edges, and in the second case we could inverse the argument to the right side. So we have decreased $e$ from a value less than $|G|$ at least by $|M_N|\frac{a-b}{2}n$ obtaining the upper bound for $e$ in this way.

If $|M_N| \ge dn$ then
$$e \le \rho_R(P_n) - dn\frac{a-b}{2}n$$
and so by \ref{lm1}
$$\rho_R(P_n) \ge \frac{n(n+1)}{2} - (\rho_R(P_n) - dn\frac{a-b}{2}n)$$
$$\rho_R(P_n) \ge \frac{n(n+1)}{4} + \frac{d(a-b)}{4}n^2$$
and the theorem holds.

So assume now that $|M_N| < dn$. Consider a subset $X$ which consists of the cliques from $N$ without the cliques from $M_N$ and with at least half of the cliques from $M$ which are not adjacent in $M$.

By $\ref{lm2}$ either the theorem holds or $|M| \ge 2n(a'-a)-cn$ and $|N| \ge 2bn-cn$. We know $|X| \ge |N|-|M_N|+|M|/2$ and so
$$|X| > (2bn-cn) -dn +\frac{2n(a'-a)-cn}{2}$$
$$|X| > (2b+(a'-a)-\frac{3}{2}c-d)n = Q$$

We will show that the subset $X$ has the chromatic number at most $Q$. At first we need to show that cliques from $M$ cannot be larger that $Q$. We know that a maximal clique in $M$ and so in $X$ can have $a'n$ vertices. So we need to show that:
$$(2b+(a'-a)-\frac{3}{2}c-d)n \ge a'n$$
and this is equivalent of
$$2b \ge a+\frac{3}{2}c+d$$
which comes from assumption (3).

Secondary we need to show that $X$ cannot have an induced clique of order greater than $Q$. Because only cliques which come from $N-M_N$ can be connected the maximal induced clique here can be of order at most $2bn$. So it is sufficient to show:
$$(2b+(a'-a)-\frac{3}{2}c-d)n \ge 2bn$$
and this is equivalent of
$$a' \ge a + \frac{3}{2}c+d$$
which comes from assumption (2).

So the subset $X$ has the chromatic number less at most $Q < |X|$ which contradicts that $G \torr P_n$, and so the theorem holds.
\end{proof}

%%%%%%%%%%%%%%%%%%%%%%%%%%%%%%%%%%%%%%%%%%%%%%%%%%%%%%%%%%%%%%%%%%%%%%%%%%%%%%%%%

% maximize d(a-b)/4 on (c>=1/14 and a+3c/2+d <= 1/2 and a+3c/2+d <= 2b and b >= 0 and a > b and d > 0 and d < 2b)
\begin{Corollary}$$\rho_R(P_n) \ge n^2/4+n^2/784+n/2.$$\end{Corollary}
\begin{proof}
Let $a' = 1/2$, $a = 1/4+1/14$, $b = 1/4$, $c = 1/14$, $d = 1/14$.
The condition (1) in \ref{th1} is clearly satisfied. Since $\frac{3}{2}c+d = 5/28$ the condition (2) $1/4+1/14+5/28 = 1/2 \le 1/2$ is satisfied and also the condition (3) $1/4+1/4+5/28 = 1/2 \le 2*1/4 = 1/2$ is satisfied. So:
$$\frac{cn(cn-1)}{4} = n^2/784 - n/56,$$
$$\frac{d(a-b)}{4}n^2 = n^2/784.$$
And by \ref{th1} we have finally $$\rho_R(P_n) \ge (n/2)(n/2+1)+n^2/784-n/56.$$
\end{proof}

%%%%%%%%%%%%%%%%%%%%%%%%%%%%%%%%%%%%%%%%%%%%%%%%%%%%%%%%%%%%%%%%%%%%%%%%%%%%%%%%%
\section{The exact value of $\rho_R$ for a double star}

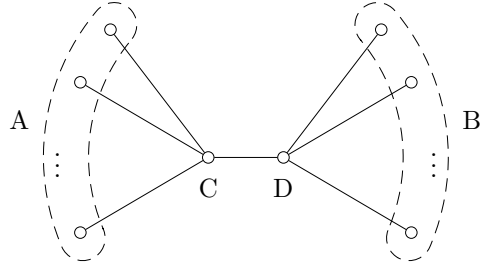
\begin{figure*}
\begin{center}
\begin{pspicture}(5,4)(0,0)
\psset{unit=1cm,linewidth=0.2pt,radius=0.08cm,labelsep=0.2cm}
\Cnode(0.7,3.7){a3}\Cnode(4.3,3.7){b3}
\Cnode(0.3,3){a2}\Cnode(4.7,3){b2}
\rput(0,2){\vdots}\Cnode(2,2){c}\Cnode(3,2){d}\rput(5,2){\vdots}
\Cnode(0.3,1){a1}\Cnode(4.7,1){b1}
\ncline{a3}{c}\ncline{a2}{c}\ncline{a1}{c}
\ncline{b3}{d}\ncline{b2}{d}\ncline{b1}{d}
\ncline{c}{d}
\nput{270}{c}{C}\nput{270}{d}{D}
\ncarcbox[linestyle=dashed,arcangle=30,boxsize=0.3,nodesep=0.3,linearc=0.3]{a3}{a1}\nbput{A}
\ncarcbox[linestyle=dashed,arcangle=30,boxsize=0.3,nodesep=0.3,linearc=0.3]{b1}{b3}\nbput{B}
\end{pspicture}
\caption{A double star $S(a,b)$ with anticliques $A$ and $B$, and central vertices $C$ and $D$.}
\end{center}
\end{figure*}

Let $S(a,b)$ be a double star of order $n = a+b+2$ with two anticliques $A$ of order $a$ and $B$ of order $b$, and two central vertices $C$ connected with $A$ and $D$ connected with $B$.

Our bounds yield:
$$\frac{(n-2)(n-1)}{2}+2 \le \rho(S(a,b)) \le \frac{(n-2)(n-1)}{2}+a+b+2.$$

\begin{Theorem}Assume that $b \ge a \ge 1$. The following holds:
$$\rho_R(S(a,b)) = \frac{(n-2)(n-1)}{2} + \left\lceil \frac{b}{a+1} \right\rceil + 3.$$\end{Theorem}
\begin{proof}
First we will show that this number of vertices is sufficient to construct a suitable $G$. We need to define orders of replication cliques of $S(a,b)$.

Let $g = \left\lceil \frac{b}{a+1} \right\rceil$. Let $A'$ be the set of replication cliques of vertices of $A$ in $G$, and $B'$ be the set of replication cliques of vertices of $B$ in $G$. Let $C'$ be the replication clique of vertex $C$ and $D'$ be the replication clique of vertex $D$. Let $c$ be the order of $C'$ and $d$ be the order of $D'$.

We get any order of the vertices $\{v_1,v_2,\ldots,v_{n-2}\}$ of $A \cup B$ such that the length of the longest sequence of consecutive vertices from $B$ is at most $g$. Each vertex $v_i$ is replicated to a clique of order $i$. So cliques in $B'$ would have orders $1,2,\ldots,g,g+2,\ldots,2g+1,2g+3,\ldots$ and cliques in $A'$ would have orders $g+1,2g+2,\ldots$. Then let $c = g+1$ and $d = 2$. We will show that the obtained graph $G$ satisfies $G \torr S(a,b)$.

Get a $k$-subset of replication cliques of $G$. Let $L$ be the set of orders of cliques from $A'$ which are in the subset, and let $R$ be the set of orders of cliques from $B'$ which are in the subset. We consider four cases depending if $C'$ or $D'$ is in the subset:
\begin{enumerate}
\item The subset does not contain $C'$ nor $D'$. So because cliques from $A' \cup B'$ have $1,2,\ldots,n-2$ orders and we have $k$ of them there is a clique of order not less than $k$.

\item The subset contains $C'$ but not $D'$. Because $C'$ is connected with the cliques from $C'$ there is a clique of order $\max(L) + g+1$ in the subset. Assume for the contrary that $k > \max(L) + g+1$. So the subset can contain all cliques of order not greater than $\max(L)$ from $A' \cup B'$, there are $\max(L)$ such cliques. Also it contains $C'$. So the remaining cliques must come from $B'$ and there must be $k-\max(L)-1$ such cliques. All of them have orders between $\max(L)+1$ and $k-1$ (inclusive) and there is exactly $k-\max(L)-1$ such cliques in $A' \cup B'$. They all must be in $B'$, but $k-\max(L)-1 > g$ and we defined orders in such a way that the longest consecutive sequences of cliques' orders from $B'$ has length $g$. So at least one clique is in $A'$ so this is impossible.

\item The subset contains $D'$ but not $C'$. The subset contains $k-1$ cliques from $A' \cup B'$ and they must be all cliques of orders $1,2,\ldots,k-1$. Since in our construction no two cliques in $A'$ have orders different by $1$ at least one of two largest cliques of orders $k-2$ and $k-1$ is in $B'$. So together with $D'$ it forms a clique of order at least $k-2+2=k$.

\item The subset contains both $C'$ and $D'$. There are $k-2$ cliques from $A' \cup B'$, so there must be a clique of order at least $k-2$ in the subset. So $\max(L) \ge k-2$ or $\max(R) \ge k-2$. In both cases there is a clique of order $\max(L) + g+1 \ge k-2 + g+1 \ge k$ or $\max(R) + 2 \ge k-2+2 = k$.
\end{enumerate}

So the constructed graph $G$ satisfies the property that for any $k$-subset there is an induced clique of order at least $k$, and so its chromatic number is at least $k$. It remains to show that any graph $G$ requires such number of vertices. So let now $G$ be any replication graph of $S(a,b)$ such that $G \torr S(a,b)$. Observe that $G$ is perfect, so each $k$-subset must have an induced clique of order at least $k$.

We sort orders of cliques from $A' \cup B'$ and obtain an orders sequence $s_1,s_2,\ldots,s_{n-2}$ of cliques $S_1,S_2,\ldots,S_{n-2}$ where $s_i \le s_j$ if $i < j$. Note that $s_i \ge i$. If $s_i - i > 0$ then we say that $S_i$ has $s_i - i$ extra vertices. Also we say that $C$ has $c-1$ extra vertices and $D$ has $d-1$ extra vertices. We will show that $G$ has at least $g+1$ extra vertices in summary.

By definition of $g$ we know that there exists $k \ge g$ such that if we get the first $k$ cliques in our sequence from $A' \cup B'$ then the last $g$ of them (which are $S_{k-g+1},\ldots,S_k$) come from $B'$. Consider a set of $g$ subsets which has for $i$: $k-g+2 \le i \le k+1$ an $i$-subset consisting of the clique $C'$ and the first $i-1$ cliques in our sequence from $A' \cup B'$. For an $i$-subset we define $L_i$ to be a set of orders of cliques which comes from $A'$ and $R_i$ to be a set of orders of cliques which comes from $B'$. So each such $i$-subset has an induced clique of order at least $i$ and so either $i \le \max(L_i)+c$ or $i \le \max(R_i)$. If $L_i$ is empty then we define $\max(L_i)=0$. In an $i$-subset the last $i-(k-g+1) \ge 1$ cliques in our order come from $B'$, so we have $\max(L_i) \le \max(R_i)$.

If for an $i$-subset there is $i \le \max(R_i)$ then $i \le s_{i-1}$. It implies that $S_{i-1}$ from $B'$ has at least one extra vertex. We consider two cases:
\begin{enumerate}
\item If all of these subsets satisfy $i \le \max(R_i)$ then there is an extra vertex in each of $S_{k-g+1},\ldots,S_k$. So we have at least $g$ extra vertices in $B'$.

If $c \ge 2$ or $d \ge 2$ then we have $g+1$ extra vertices in summary. Otherwise consider any $r$-subset ($2 \le r \le n$) consisting of the first $r-2$ cliques in our sequence from $A' \cup B'$ and the cliques $C'$ and $D'$. The largest clique of order at least $r$ must consist of $S_{r-2}$ and either $C'$ or $D'$, so $s_{r-2}+1 \ge r$. So $S_{r-2}$ has at least one extra vertex. Because this holds for any such $r$-subset any clique from $A' \cup B'$ has an extra vertex, and because $a+b > g$ we have at least $g+1$ extra vertices.

\item If for one of these subsets $i > \max(R_i)$ then it must satisfy $i \le \max(L_i)+c$. Consider the largest $i$-subset of them. So we have $k-i+1$ extra vertices in $S_i,\ldots,S_k$ from $B'$ implied by the subsets larger than $i$.

If $L_i$ is empty then we have $i \le c$ and so $C'$ has $i-1$ extra vertices. Together with the $k-i+1$ extra vertices from $B'$ we have $k \ge g$ extra vertices in summary in $B'$ and $C'$.

If $L_i$ is non-empty then let $\max(L_i) = s_j$. $S_j \in A'$ and so $j \le k-g$ because the cliques $S_{k-g+1},\ldots,S_k$ are in $B'$. Having $i \le s_j+c$ we could write $i \le j+e_j+c$ where $e_j$ is the number of extra vertices in $S_j$. From these we obtain $i \le (k-g)+e_j+c$ and so $e_j+c-1 \ge i-k+g-1$. Obtained $e_j+c-1$ is the number of extra vertices in both $S_j \in A'$ and $C'$. If we add the $k-i+1$ extra vertices from $B'$ then we have at least $(i-k+g-1)+(k-i+1)=g$ extra vertices in summary in $A' \cup B'$ and $C'$.

If $d \ge 2$ then we are done, so assume that $d = 1$.

If we add $D'$ to the $i$-subset from this case and obtain an $(i+1)$-subset then it must have an induced clique of order at least $i+1$. So we have the three following subcases depending on what the clique of order at least $i+1$ consist of:
\begin{enumerate}
\item If the clique consists of $D'$ and some clique from $B'$ then $i+1 = max(R_i)+1$ and we have a contradiction as we assumed $i > \max(R_i)$. 

\item If the clique consists of $C'$ and some clique from $A'$ (or just $C'$ if $L_i$ is empty) then $i+1 = \max(L_i)+c$ and we could go similar as in the case:

If $L_i$ is empty then $i+1 = c$ and the $i$ extra vertices in $C'$ together with the $k-i+1$ extra vertices from $B'$ give $k+1 \ge g+1$ extra vertices in summary.

If $L_i$ is non-empty then we could write $i+1 \le s_j+c$ obtaining $(i+1-k+g-1)+(k-i+1)=g+1$ extra vertices in summary in $A' \cup B'$ and $C'$.

\item If the clique consists of $C'$ and $D'$ then $i+1 \le c+1$ and so $i \le c$.

If $L_i$ is non-empty then $i+1 \le \max(L_i)+c$ and we could follow the previous subcase.

Assume that $L_i$ is empty. $C'$ has $i-1$ extra vertices and together with the $k-i+1$ extra vertices from $B'$ we have $k$ extra vertices in summary in $C'$ and $B'$.

If $k > g$ then we are done, so assume that $k = g$. Note that we have counted extra vertices from $B'$ which are only in the first $g$ cliques in our sequence from $A' \cup B'$, also we have counted only one extra vertex in one clique of $B'$.

If $a=1$ and $b=1$ then $S(1,1)$ is $P_4$ and the theorem is true. Otherwise there are at least $g+2$ cliques in $A' \cup B'$. We will show that there is an extra vertex in the cliques $S_{n-3}$ or $S_{n-2}$ which are not in the first $g$ cliques in our sequence, or in some clique in $A' \cup B'$ there are two extra vertices.

Get a $n$-subset consisting of all cliques. There are three sub-subcases depending where is the clique of order at least $n$.
\begin{itemize}
\item If $c+1 \ge n$ then $c$ has $n-2 > g$ extra vertices.

\item If $s_{n-2}+1 \ge n$ (when $S_{n-2} \in B'$) then $S_{n-2}$ has one extra vertex.

\item If $s_{n-2}+c \ge n$ (when $S_{n-2} \in A'$) then we remove $C'$ from the subset and obtain either $s_{n-2} \ge n-1$ or $s_j+1 \ge n-1$ for some $j \le n-3$ ($S_j \in B'$). In the first case we are done. In the second case we have $s_j \ge n-2$ and so $S_j \in B'$ has at least one extra vertex if $j = n-3$ or has at least two extra vertices if $j < n-3$.
\end{itemize}

\end{enumerate}

\end{enumerate}

So in all cases we have at least $g+1$ extra vertices and the theorem holds.

\end{proof}

%%%%%%%%%%%%%%%%%%%%%%%%%%%%%%%%%%%%%%%%%%%%%%%%%%%%%%%%%%%%%%%%%%%%%%%%%%%%%%%%%
\section{Experiments}

\begin{Lemma}For a given graph $G$ which is a replication graph of $P_n$, the problem of verifying if $G \torr P_n$ can be solved in time $O(n^2)$ and memory $O(n)$.\end{Lemma}

By computer search, we have found all minimal order replication graphs for each path up to 16 vertices and so we have the exact value of $\rho_R(P_n)$ for $n \leq 16$. These shows that the upper bound is tend to be very close to the exact value, especially we have the conjectures:

\begin{Conjecture}For odd $n \ge 1$ the upper bound is tight:
$$\rho_R(P_n) = \left\{\begin{array}{ll}
n^2/4+n^2/16+3n/8+5/16 & \mbox{if $n \equiv 1 \mod 4$}\\
n^2/4+n^2/16+3n/8+1/16 & \mbox{if $n \equiv 3 \mod 4$}\\
\end{array}\right.$$
\end{Conjecture}

\begin{Conjecture}For any $n \ge 1$ there is $\rho_R(P_n) = n^2/4+n^2/16+O(n)$.\end{Conjecture}

Here are the exact values, the number of minimal order replication graphs and representations of some of them.

{\small\twocolumn
$\rho_R(P_5) = 10 (5 graphs)\\
\begin{array}{rlllll}
1)&1&2&2&2&3\\
2)&1&2&3&2&2\\
3)&2&1&3&2&2\\
4)&2&2&1&2&3\\
5)&3&1&1&2&3\\
\end{array}$\vs

$\rho_R(P_6) = 14 (14 graphs)\\
\begin{array}{rllllll}
1)&1&2&3&2&2&4\\
2)&1&2&3&4&2&2\\
3)&1&3&2&2&3&3\\
4)&1&3&3&2&2&3\\
5)&1&3&3&3&2&2\\
6)&2&1&3&2&2&4\\
7)&2&2&3&1&2&4\\
8)&2&2&3&1&3&3\\
9)&2&2&3&3&2&2\\
10)&2&2&4&1&2&3\\
11)&2&3&3&1&2&3\\
12)&2&3&3&2&1&3\\
13)&3&2&1&1&3&4\\
14)&3&2&1&2&2&4\\
\end{array}$\vs

$\rho_R(P_7) = 18 (3 graphs)\\
\begin{array}{rlllllll}
1)&1&2&3&4&2&2&4\\
2)&1&3&2&2&4&3&3\\
3)&2&3&3&2&1&3&4\\
\end{array}$\vs

$\rho_R(P_8) = 23 (1 graph)\\
\begin{array}{rllllllll}
1)&1&4&4&2&2&4&3&3\\
\end{array}$\vs

$\rho_R(P_9) = 29 (20 graphs)\\
\begin{array}{rlllllllll}
1)&1&3&3&3&4&5&2&3&5\\
2)&1&3&4&5&2&3&3&3&5\\
3)&1&3&4&5&2&3&5&3&3\\
4)&1&3&4&5&3&2&5&3&3\\
5)&1&4&3&2&4&4&5&3&3\\
6)&1&4&4&2&2&4&3&3&6\\
7)&1&5&3&3&4&2&2&4&5\\
8)&1&5&4&2&2&4&3&3&5\\
9)&2&2&4&3&3&5&1&4&5\\
10)&2&2&4&5&1&4&3&3&5\\
11)&2&3&5&4&3&1&5&3&3\\
12)&3&2&5&4&3&1&5&3&3\\
13)&3&3&1&3&4&5&2&3&5\\
14)&4&2&1&3&4&5&2&3&5\\
15)&4&3&1&2&4&5&2&3&5\\
16)&4&3&1&5&4&2&2&3&5\\
17)&4&3&2&4&4&2&1&4&5\\
18)&4&3&2&4&5&1&2&3&5\\
19)&5&2&2&4&4&2&1&4&5\\
20)&5&3&2&1&5&4&2&2&5\\
\end{array}$\vs

$\rho_R(P_{10}) = 35 (3 graphs)\\
\begin{array}{rllllllllll}
1)&1&3&4&5&2&3&5&6&3&3\\
2)&1&5&4&2&2&5&5&3&3&5\\
3)&5&3&2&5&4&3&1&2&4&6\\
\end{array}$\vs

$\rho_R(P_{11}) = 42 (16 graphs)\\
\begin{array}{rlllllllllll}
1)&1&3&4&5&2&3&5&6&3&3&7\\
2)&1&3&4&5&2&3&5&6&3&4&6\\
3)&1&3&4&5&3&3&5&6&2&4&6\\
4)&1&3&4&6&2&3&5&6&3&3&6\\
5)&1&3&5&5&2&3&5&6&3&3&6\\
6)&1&3&5&5&3&2&5&6&3&3&6\\
7)&1&4&5&6&3&3&4&2&6&4&4\\
8)&1&6&4&4&5&2&3&5&6&3&3\\
9)&1&6&5&2&2&4&5&5&3&3&6\\
10)&1&6&5&2&3&5&4&4&6&3&3\\
11)&2&2&5&5&2&3&5&6&3&3&6\\
12)&3&4&5&2&3&5&5&3&1&5&6\\
13)&3&4&5&2&3&5&6&2&2&4&6\\
14)&3&5&5&3&2&5&4&3&1&5&6\\
15)&4&4&5&3&2&5&5&2&1&5&6\\
16)&4&4&5&3&2&5&6&1&2&4&6\\
\end{array}$\vs

$\rho_R(P_{12}) = 49 (1 graph)\\
\begin{array}{rllllllllllll}
1)&1&6&5&2&3&6&6&3&3&6&4&4\\
\end{array}$\vs

$\rho_R(P_{13}) = 58 (116 graphs)$\\

$\rho_R(P_{14}) = 66 (3 graphs)\\
\begin{array}{rllllllllllllll}
1)&1&4&5&7&2&4&6&7&3&4&7&8&4&4\\
2)&1&7&5&3&2&7&6&3&3&7&7&4&4&7\\
3)&1&7&6&2&3&6&6&3&3&7&7&4&4&7\\
\end{array}$

}\onecolumn

We have observed also an interesting property about the cycles and state the following conjecture:

\begin{Conjecture}$\rho_R(C_n) = \rho_R(P_n)$ for $n \ge 6$.\end{Conjecture}

We note that $\rho_R(C_4) = 6$, $\rho(C_5) = 9$. The conjecture was checked for $n \le 12$.

%%%%%%%%%%%%%%%%%%%%%%%%%%%%%%%%%%%%%%%%%%%%%%%%%%%%%%%%%%%%%%%%%%%%%%%%%%%%%%%%%

%%%%%%%%%%%%%%%%%%%%%%%%%%%%%%%%%%%%%%%%%%%%%%%%%%%%%%%%%%%%%%%%%%%%%%%%%%%%%%%%%
\end{document}